\documentclass[reqno]{amsproc}

\usepackage{enumitem}
\usepackage{bbm}
\usepackage{geometry}
\usepackage{amsmath,amsthm}
\usepackage{color}
\usepackage{hyperref}
\definecolor{myurlcolor}{rgb}{0,0,0.7}
\hypersetup{colorlinks,
linkcolor=myurlcolor,
citecolor=myurlcolor,
urlcolor=myurlcolor}

\theoremstyle{plain}
\newtheorem{thm}{Theorem}
\newtheorem*{thm*}{Theorem}
\newtheorem{lem}[thm]{Lemma}

\theoremstyle{definition}

\theoremstyle{remark}

\newcommand{\Z}{\mathbb{Z}}
\newcommand{\N}{\mathbb{N}}

\newcommand{\Q}{\mathbb{Q}}
\newcommand{\R}{\mathbb{R}}
\newcommand{\ra}{\rightarrow}
\newcommand{\eq}[1]{~(\ref{#1})}
\newcommand{\bydef}{\stackrel{\mathrm{def}}{=}}
\newcommand{\tr}{\mathrm{tr}}
\renewcommand{\H}{\mathcal{H}}
\newcommand{\U}{\mathcal{U}}
\newcommand{\spec}{\mathrm{spec}}

\newcommand{\eps}{\varepsilon}

\renewcommand{\labelenumi}{(\alph{enumi})}

\newcommand{\beq}{\begin{equation}}
\newcommand{\eeq}{\end{equation}}

\begin{document}

\title[On infinite-dimensional state spaces]{On infinite-dimensional state spaces}

\author{Tobias Fritz}
\address{Perimeter Institute for Theoretical Physics\\
Waterloo N2L 2Y5, Canada}
\email{tfritz@perimeterinstitute.ca}

\thanks{\textit{Acknowledgements.} Discussions with Antonio Ac\'in, Toby Cubitt, Anthony Leverrier, Magdalena Musat, Miguel Navascu{\'e}s, Mikael R{\o}rdam, Andreas Thom and others are much appreciated and have greatly helped in interpreting the present results. Financial support has been granted by the EU STREP QCS. This research was also supported in part by Perimeter Institute for Theoretical Physics.}

\date{\today}

\begin{abstract}
It is well-known that the canonical commutation relation $[x,p]=i$ can be realized only on an infinite-dimensional Hilbert space. While any finite set of experimental data can also be explained in terms of a finite-dimensional Hilbert space by approximating the commutation relation, Occam's razor prefers the infinite-dimensional model in which $[x,p]=i$ holds on the nose. This reasoning one will necessarily have to make in any approach which tries to detect the infinite-dimensionality. One drawback of using the canonical commutation relation for this purpose is that it has unclear operational meaning. Here, we identify an operationally well-defined context from which an analogous conclusion can be drawn: if two unitary transformations $U,V$ on a quantum system satisfy the relation $V^{-1}U^2V=U^3$, then finite-dimensionality entails the relation $UV^{-1}UV=V^{-1}UVU$; this implication strongly fails in some infinite-dimensional realizations. This is a result from combinatorial group theory for which we give a new proof. This proof adapts to the consideration of cases where the assumed relation $V^{-1}U^2V=U^3$ holds only up to $\eps$ and then yields a lower bound on the dimension.
\end{abstract}

\maketitle

\section{Witnessing infinite-dimensionality?}

Standard quantum field theory posits that the state space of a quantum field---i.e.,~Fock space---is a Hilbert space of countably infinite dimension. On the other hand, modern quantum information theory concerns itself mostly with state spaces of finite dimension. One reason for this may be that in practical applications, the subspace of experimentally accessible states is typically finite-dimensional. Also, it is sometimes suggested that the state space of any quantum system (of bounded size) may be finite-dimensional; see e.g.~\cite{BiSu,DA}.

Let us consider a quantum system containing observables $x$ and $p$ satisfying the canonical commutation relation (CCR),
\beq
\label{ccr}
[x,p]=i\mathbbm{1} .
\eeq
We do not want to assume a specific r\^ole or interpretation of these observables; in particular, we do not assume them to be bounded or unbounded operators. In finite dimensions, satisfying the CCR is impossible, as can be seen by applying the trace to\eq{ccr} and noting that the trace of a commutator always vanishes, while the trace of the right-hand side does not. In particular, having a quantum system together with a pair of observables which is guaranteed to satisfy the CCR is necessarily infinite-dimensional. In this sense, the CCR can witness infinite-dimensionality. While finite-dimensional approximate representations of the CCR exist achieving any desired degree of accuracy, Occam's razor clearly favors the infinite-dimensional description: if a certain simple relation between the observables seems to hold up to experimental accuracy, it would be very unnatural to describe the system by a model in which this relation would only hold up to some error smaller than the experimental error. Or, in Einstein's words~\cite{Einstein},
\begin{quote}
It can scarcely be denied that the supreme goal of all theory is to make the irreducible basic elements as simple and as few as possible without having to surrender the adequate representation of a single datum of experience.
\end{quote}
If experimental data indicates a certain relation to hold up to experimental error, then any model in which this relation holds exactly should be preferred over one in which the relation does not hold exactly. In accordance with the scientific method~\cite{popper}, this holds until new data appears which may refute the relation.

Similar considerations necessarily apply to any idea witnessing infinite-dimensionality, and in particular to the one which we are going to propose. Any such idea needs to be based on a set of certain relations between observables or other operators with operational significance, such that these relations are compatible only on an infinite-dimensional Hilbert space. In this way, one could try and witness infinite-dimensionality by probing these relations experimentally. However, any infinite-dimensional representation of these relations can be compressed to a finite-dimensional subspace\footnote{Upon writing $P$ for the projection onto the subspace, this means replacing any operator $X$ by $PXP$, which is an operator acting on the subspace.}; such a compression yields an approximate finite-dimensional representation of the relations. Upon making these subspaces larger and larger, this approximate representation can achieve any desired degree of accuracy. This is one reason why witnessing infinite-dimensionality with perfect certainty is impossible; one needs to appeal to Occam's razor and say that any \emph{reasonably nice} model reproducing the experimental data requires an infinite-dimensional Hilbert space.

In principle, it is conceivable that some of the tested relations are not only based on experimental input, but also on other physical principles, for example on the commutativity of observables located on spacelike separated spacetime regions. Ultimately, the same reservations apply in this case: one can find models in finite Hilbert space dimension in which the physical principle, and therefore the enforced relations, only hold approximately. Since any physical principle itself is ultimately based on experiments, this is qualitatively the same situation as in the previous paragraph. The only advantage of this approach is quantitative and lies in the fact that a general and well-established physical principle would typically be based on significantly more experimental data than, say, some exotic relation pertaining to a single type of physical system.

In conclusion, we find that the CCR\eq{ccr} witnesses infinite-dimensionality in the sense that either (i) one works with infinite-dimensional Hilbert space, or (ii) one works in finite dimension, but has an unnatural and contrived model making slightly different predictions.

\section{Main result}

Here, we would like to present another set of relations besides\eq{ccr} which is, in the same sense, compatible only with an infinite-dimensional Hilbert space. The advantage of our proposal over the CCR is that it is (almost) device-independent, which makes it operationally well-defined. Our observation is this:

\begin{thm}
\label{mainthm}
\begin{enumerate}
\item\label{mainthm1} If $\H$ is finite-dimensional and $U,V\in\U(\H)$ satisfy $V^{-1}U^2V=U^3$, then $UV^{-1}UV=V^{-1}UVU$.
\item\label{mainthm2} If $\H$ is infinite-dimensional, then there are $U,V\in \U(\H)$ which satisfy $V^{-1}U^2V=U^3$, but $\langle UV^{-1}UV\psi,\, V^{-1}UVU\psi\rangle=0$ for some appropriate $\psi\in\H$. 
\end{enumerate}
\end{thm}

This amounts to a big gap between infinite dimensions and any number of finite dimensions: in the finite-dimensional case, the relation $V^{-1}U^2V=U^3$ implies $\langle UV^{-1}UV\psi,\, V^{-1}UVU\psi\rangle=1$ for any $\psi$, while in infinite dimensions, also $\langle UV^{-1}UV\psi,\, V^{-1}UVU\psi\rangle=0$ is possible for some $\psi$.

The theorem essentially follows from the known proofs~\cite{M} that the Baumslag-Solitar group~\cite{BS}
\beq
\label{BS23}
BS(2,3)=\langle u,v \:|\: v^{-1}u^2v=u^3 \rangle
\eeq
is finitely generated but not residually finite, and therefore not maximally almost periodic~\cite[Prop.~4]{HRGV}. This means the following: the group\eq{BS23} contains non-trivial elements, in this case e.g. $(uv^{-1}uv)(v^{-1}uvu)^{-1}=uv^{-1}uvu^{-1}v^{-1}u^{-1}v$, which act trivially in every finite-dimensional unitary representation. In fact, any finitely presented but not residually finite group yields an analogous theorem; therefore, literature like~\cite{B,H1,H2,CFP} produces a myriad of examples with analogous implications for quantum mechanics. Also, $C^*$-algebraic results of similar flavor exist~\cite{Bekka}.

In Appendix~\ref{proofs}, we give a new proof of Theorem~\ref{mainthm} which is independent of the literature on group theory and uses nothing but linear algebra. 

Applying Theorem~\ref{mainthm} in order to witness infinite-dimensionality, in the same sense as above, is simple. We now explain how to do this and then comment on certain specific aspects of the procedure.

First of all, one needs to have a quantum system for which the theory predicts an infinite-dimensional state space; for example a mode of the electromagnetic field. Furthermore, one needs experimental access to unitary transformations $U$ and $V$ and their inverses $U^{-1}$ and $V^{-1}$ for which the theory predicts the relation $V^{-1}U^2V=U^3$, as well as a state $\psi$ for which $\langle UV^{-1}UV\psi,\, V^{-1}UVU\psi\rangle=0$ holds\footnote{This latter condition can certainly be relaxed, but having orthogonality gives the most dramatic effect.}. The existence of such $U$, $V$ and $\psi$ is guaranteed by Theorem~\ref{mainthm}\ref{mainthm2}.

Then the experiment should do the following:

\renewcommand{\labelenumi}{(\roman{enumi})}
\begin{enumerate}
\item Verify the relation $V^{-1}U^2V=U^3$ up to phase by generating two copies of an initial state $\phi$, applying the composed transformation $V^{-1}U^2V$ to the first copy and $U^3$ to the second copy, and then measure whether the two resulting states coincide up to phase,
\beq
\label{check0}
|\langle V^{-1}U^2V\phi, U^3\phi\rangle|^2\stackrel{?}{=}1 .
\eeq
One way to compare these two states is to perform a swap test~\cite{swap} and repeat many runs in order to gather statistics. This procedure should be repeated for as many initial states $\phi$ as are accessible to experimental preparation.
\item Prepare two copies of the initial state $\psi$ from above, apply the composed transformation $UV^{-1}UV$ to the first and $V^{-1}UVU$ to the second, and perform the swap test in order to determine whether
\beq
\label{check1}
\langle UV^{-1}UV\psi, V^{-1}UVU\psi\rangle \stackrel{?}{=} 0 .
\eeq
\end{enumerate}

If\eq{check0} holds for all $\phi\in\H$, then there is a phase $e^{i\alpha}$ such that $U^{-3}V^{-1}U^2V=e^{i\alpha}\cdot\mathbbm{1}$. Setting $U'\bydef e^{i\alpha}U$ then gives a pair of unitaries $(U',V)$ such that
$$
V^{-1}U'^2V=U'^3 ,\qquad U'V^{-1}U'V\psi \perp V^{-1}U'VU'\psi .
$$
Theorem~\ref{mainthm} says that no quantum-mechanical model on a finite-dimensional Hilbert space can reproduce these theoretical predictions. In other words, any finite-dimensional model needs to relax the relation $V^{-1}U^2V=e^{i\alpha} U^3$ in one way or the other. Given that all experimental evidence would suggest this relation to hold, any such model would be unnatural and contrived.

Of course, in accordance with the scientific method~\cite{popper}, ``verifying'' a relation like $V^{-1}U^2V=U^3$ or like\eq{ccr} means repeatedly putting it to test on an ever increasing number of initial states and continuously improving statistics without finding a violation. Both due to the very property of infinite-dimensionality and due to the perfect accuracy and infinite statistics required, \emph{proving} the relation $V^{-1}U^2V=U^3$ is a matter of impossibility; the equation\eq{check0} can at most be verified for finitely many $\phi\in\H$. This equation, after being predicted by a theoretical model, can at most resist all attempts at experimental falsification. The impossibility of \emph{proving} relations is a generic issue which arises not only in any proposal of witnessing infinite-dimensionality, but rather pertains to any empirical science.

Since we take $U$ and $V$ to stand for the unitary operators which are part of the infinite-dimensional model predicted by the theory, rather than for the actual transformation carried out on the system---which are not going to be perfect unitaries anyway, but rather quantum operations---there is no need to include error terms in the theoretical predictions\eq{check0} and\eq{check1}.

Let us compare this method of witnessing infinite-dimensionality with doing it via testing the CCR\eq{ccr}. What does it mean to say that two observables $x$ and $p$ satisfy the CCR? Is there a device-independent test for the CCR which would thereby witness the infinite-dimensionality of the state space? Although the CCR has been put to experimental test~\cite{ZPK} (on a relatively small set of initial states), this rests on many assumptions about theoretical models for quantum optics and therefore is not device-independent. Our proposal is superior in this respect since it has clear operational meaning, and is almost device-independent in the sense that the only ``device dependence'' enters through the use of a swap test.

On the other hand, while infinite-dimensional quantum system equipped with observables satisfying the CCR are at the very heart of quantum mechanics and quantum field theory, we have not yet been able to think of any physical system which allows the experimental implementation of unitaries $U$ and $V$ satisfying $V^{-1}U^2V=U^3$, but not $V^{-1}UVU=UV^{-1}UV$. For example, even though linear optics on finitely many modes is concerned with infinite-dimensional state spaces, we expect that taking $U$ and $V$ to be given by Bogoliubov transformations is not sufficient, since a Bogoliubov transformation is a symplectic matrix of finite size.

\section{A quantitative version}
\label{quant}

For any $d\in\N$, we write $N_d$ for a certain integer defined in terms of a least common multiple,
$$
N_d\bydef 3^d \cdot \mathrm{lcm}\left\{3^1-2^1,\ldots,3^d-2^d\right\} .
$$
It grows exponentially in $d^2$.

We have seen so far that any finite-dimensional model which tries to reproduce the predictions of the infinite-dimensional one will need to relax the equation $V^{-1}U^2V=U^3$. Now we give a quantitative version of Theorem~\ref{mainthm} which gives a bound on the Hilbert space dimension required as a function of $||V^{-1}U^2V-U^3||$.

\begin{thm}
\label{appthm}
Let $d\geq 3$. If $\dim(\H)\leq d$ and $U,V\in\U(\H)$ satisfy
\beq
\label{cons}
||V^{-1}U^2V-U^3|| < \eps
\eeq
for some $\eps<\frac{1}{6\cdot 3^d d N_d}$, then 
\beq
\label{ass}
||UV^{-1}UV-V^{-1}UVU|| < 4d^3 N_d\eps.
\eeq
\end{thm}

See also Appendix~\ref{proofs} for the proof. The counterpoint showing that the same conclusion does not hold for more than $d$ dimensions is still Theorem~\ref{mainthm}.\ref{mainthm2}. 

The physical interpretation of this result is similar to the previous one. In infinite dimensions, there are situations in which $V^{-1}U^2V=U^3$ holds, while\eq{ass} does not. Any quantum-mechanical model which violates\eq{ass} and deviates from the relation $V^{-1}U^2V=U^3$ by at most $\tfrac{1}{6\cdot 3^d d N_d}$ will require a Hilbert space of dimension greater than $d$. Again, the usual leap of faith is required here, or rather appeal to the scientific method~\cite{popper}: since the hypothesis $||V^{-1}U^2V-U^3||<\eps$ cannot be verified on all initial states, it has to be taken as a working hypothesis which gets constantly subjected to experimental scrutiny and is regarded as valid as long as evidence to the contrary is found.

Since the growth of $N_d$ as a function of $d$ is so huge, the dimension bounds obtained by this method are extremely weak.

\bibliographystyle{plain}
\bibliography{infinite_witness}

\newpage

\appendix

\section{Proofs}
\label{proofs}

By abuse of terminology, we define an eigenvalue $\lambda$ of a unitary $U\in\U(\H)$ to be a number $\lambda\in\R$ for which there exists a non-zero vector $\xi\in\H$ such that
$$
U\xi = e^{2\pi i\,\lambda}\xi \:.
$$
This $\lambda$ is well-defined up to addition of an integer. Hence we think of $\lambda$ as an element of the abelian group $\R/\Z$. The spectrum $\spec(U)$ is then a subset of $\R/\Z$. If $\lambda$ is an eigenvalue of $U$ and $n\in\Z$, then $n\lambda$ is an eigenvalue of $U^n$.

We define the absolute value function
$$
|\cdot| \::\: \R/\Z \longrightarrow \R \:,\qquad [x]\mapsto \min_{n\in\Z} |x+n| \:.
$$
where $x\in\R$ is some representative of an equivalence class $[x]\in\R/\Z$. By this definition, the absolute value of $\lambda\in\R/\Z$ always satisfies $0\leq |\lambda|\leq \tfrac{1}{2}$. Measuring the distance between $\lambda\in\R/\Z$ and $\lambda'\in\R/\Z$ by $|\lambda-\lambda'|$ defines a metric on $\R/\Z$. Moreover, the triangle inequality in the form $|\lambda+\lambda'|\leq |\lambda|+|\lambda'|$ is also valid.

\begin{proof}[Proof of Theorem~\ref{mainthm}]
\begin{enumerate}[leftmargin=0cm,itemindent=.8cm]
\item[\ref{mainthm1}] We write $d=\dim(\H)$.

We start by noting that the assumption $V^{-1}U^2V=U^3$ means that $U^2$ and $U^3$ are conjugate. In particular, $\spec(U)$ has the property that for all $\lambda_0\in \spec(U)$ there exists some $\lambda_1\in \spec(U)$ such that
\beq
\label{maineq}
3\lambda_0=2\lambda_1 \qquad\textrm{(in $\R/Z$)}.
\eeq
Conversely, for any $\lambda_1\in\spec(U)$ there exists some $\lambda_0\in\spec(U)$ such that\eq{maineq} holds. Applying the first of these observations $d$ times starting with any $\lambda_0\in\spec(U)$ yields a sequence $\lambda_0,\ldots,\lambda_d$ with $3\lambda_k=2\lambda_{k+1}$, which implies $3^n\lambda_k=2^n\lambda_{k+n}$ for all sensible $k$ and $n$. By $|\spec(U)|\leq d$ and the pigeonhole principle, there are indices $k$ and $n\geq 1$ such that $\lambda_k=\lambda_{k+n}$, which implies
$$
\left(3^n-2^n\right)\lambda_k=0 .
$$
In other words, $\lambda_k$ can be written as a rational number in $\Q/\Z$ with denominator $3^n-2^n$. By $3^k\lambda_0=2^k\lambda_k$, we find
$$
3^k(3^n-2^n)\lambda_0 = 0 ,
$$
so that $\lambda_0$ is a rational number in $\Q/\Z$ with denominator $3^k(3^n-2^n)$. Since this number is a divisor of $N_d$ (as defined in Section~\ref{quant}) and $\lambda_0\in\spec(U)$ was arbitrary, we find that any eigenvalue of $U$ satisfies $N_d\lambda=0$. Since $N_d$ is odd, this proves that if $\lambda\in\spec(U)$, then $\left(\lambda+\tfrac{1}{2}\right)\not\in\spec(U)$. In other words, $U^2$ does not have more spectral degeneracy than $U$, and every eigenvector of $U^2$ is also an eigenvector of $U$. Similarly, every eigenvector of $V^{-1}U^2V$ is also an eigenvector of $V^{-1}UV$.

Now since $V^{-1}U^2V=U^3$ commutes with $U$, there is a basis of common eigenvectors of $V^{-1}U^2V$ and $U$. By the result of the previous paragraph, this basis is also a basis of common eigenvectors for $V^{-1}UV$ and $U$, so that these two commute.

\item[\ref{mainthm2}] It is enough to show this for some specific infinite-dimensional separable $\H$.

We consider the set $\N\times\Z$ equipped with bijections $U$ and $V$. For $U$ we take the shift
\beq
\label{UVex}
U\::\:  (x,y)\mapsto (x,y+1) .
\eeq
Then the orbits of $U^2$ (resp.~$U^3$) are the sets of the form $(x,y+2\Z)$ (resp.~$(x,y+3\Z)$). In order to define $V$, we choose any bijection from the orbits of $U^3$ to the orbits of $U^2$; since there are countably many of each kind, this is certainly possible. Any such bijection can be lifted to a bijection $V:\N\times\Z\mapsto\N\times\Z$ with $V^{-1}U^2V=U^3$. With the appropriate choices, this can be done such that
$$
V(0,0) = (0,0),\qquad V(0,1)=(0,1),\qquad V(0,-2) = (0,1).
$$
since $(0,-2)$ and $(0,1)$ are in the same $U^3$-orbit, while $(0,-1)$ and $(0,1)$ are in the same $U^2$-orbit. With such a choice of $V$, we obtain $UV^{-1}UV(0,0)=(0,-1)$, while $V^{-1}UVU(0,0)=(0,0)$. In particular, $UV^{-1}UV\neq V^{-1}UVU$. 

Linearly extending these bijections to unitary operators $U$ and $V$ on $\ell^2(\N_0\times\Z)$ preserves these properties. One can then take $\psi$ to be the basis vector associated to $(0,0)$, and the equation $\langle UV^{-1}UV\psi,\, V^{-1}UVU\psi\rangle=0$ holds.
\end{enumerate}
\end{proof}

Now we turn towards the proof of Theorem~\ref{appthm}, beginning with a lemma.

\begin{lem}
\label{evlemma}
Let $S\in \U(\H)$ be a unitary with $\H$ finite-dimensional. If $\xi\in\H$ is a unit vector and $\beta\in\R/\Z$ is such that
$$
||S\xi - e^{2\pi i\,\beta}\xi|| < \delta \:,
$$
for some $0<\delta<1$, then there is an eigenvalue $\lambda\in\spec(S)$ with $|\lambda-\beta|<\delta$.
\end{lem}

\begin{proof}
We assume $\beta=0$ without loss of generality. Let
$$
S=\sum_{\lambda\in\spec(S)} e^{2\pi i\, \lambda}P_\lambda
$$
be the spectral decomsposition of $S$. Then the assumption means that
$$
||S\xi-\xi||=\left|\left|\sum_{\lambda} \left(e^{2\pi i\,\lambda}-1\right) P_\lambda \xi \right|\right| \leq \sqrt{\sum_\lambda \left|e^{2\pi i\,\lambda}-1\right|^2\cdot ||P_\lambda\xi||^2} < \delta \:.
$$
Since $\sum_\lambda ||P_\lambda\xi||^2=1$, this means that there is at least one $\lambda\in\spec(S)$ with $|e^{2\pi i\,\lambda}-1|<\delta$. Since $|\lambda|<|e^{2\pi i\,\lambda}-1|$ for $|\lambda|\leq\tfrac{1}{2}$, this eigenvalue satisfies $|\lambda|<\delta$.
\end{proof}

\begin{proof}[Proof of Theorem~\ref{appthm}]
We partly imitate the proof of Theorem~\ref{mainthm}.\ref{mainthm1}, this time keeping track of more quantitative estimates.

Given the assumption, it needs to be shown that the commutator $[V^{-1}UV,\, U]$ is small. 

The assumption\eq{ass} means that $U^2$ and $U^3$ are close to conjugate. In particular, applying Lemma~\ref{evlemma} to $S=V^{-1}U^2V$ with $\xi$ any eigenvector of $U$ associated to an eigenvalue $\lambda_0\in\spec(U)$ shows the existence of some $\lambda_1\in\spec(U)$ such that
\beq
\label{maineq2}
|2\lambda_1-3\lambda_0|<\eps .
\eeq
Conversely, analogous reasoning shows that for any $\lambda_1\in\spec(U)$ there exists some $\lambda_0\in\spec(U)$ such that\eq{maineq2} holds. Applying the first of these observations $d$ times starting with any $\lambda_0\in\spec(U)$ yields a sequence $\lambda_0,\ldots,\lambda_d$ with $|3\lambda_k-2\lambda_{k+1}|<\eps$, which implies $|3^n\lambda_k-2^n\lambda_{k+n}|<3^{n-1}n\eps$ for all sensible $k$ and $n$.

By $|\spec(U)|\leq d$ and the pigeonhole principle, there are indices $k$ and $n\geq 1$ such that $\lambda_k=\lambda_{k+n}$, which implies
$$
\left|\left(3^n-2^n\right)\lambda_k\right| < 3^{n-1} n\eps .
$$
By $\left|3^k\lambda_0-2^k\lambda_k\right| < 3^{k-1} k\eps$, we find
\begin{align}
\left|3^k(3^n-2^n)\lambda_0 \right| &= \left| 3^n\left(3^k\lambda_0 - 2^k\lambda_k\right) + 2^k\left(3^n-2^n\right)\lambda_k + 2^n\left(2^k\lambda_k-3^k\lambda_0\right) \right| \\
&< 3^n\cdot 3^{k-1}k\eps + 2^k\cdot 3^{n-1}n\eps + 2^n\cdot 3^{k-1}k\eps <  3^d d\eps \:,
\end{align}
where the last estimate also used $n+k\leq d$. This shows that $\lambda_0$ is close to a rational number with denominator $3^k(3^n-2^n)$. In order to remove the dependence on $k$ and $n$, we recall that $3^k(3^n-2^n)$ divides $N_d$, so that the estimate can be weakened to
$$
|N_d\lambda_0| < 3^d d N_d \eps .
$$
Since $3^d d N_d \eps < 1/6$ by assumption, the distance from $\lambda_0$ to the closest rational with denominator $N_d$ is less than $\tfrac{1}{6N_d}$. We denote the associated integer numerator by $\mu_0$, so that
\beq
\label{approx}
\left|\lambda_0 - \frac{\mu_0}{N_d}\right| < 3^d d \eps < \frac{1}{6N_d}
\eeq
holds for any $\lambda_0\in\spec(U)$. For two different eigenvalues $\lambda,\lambda'\in\spec(U)$, the associated numerators $\mu,\mu'$ coincide as soon as $|\lambda-\lambda'| < \tfrac{1}{3N_d}$. Having the same associated numerator is an equivalence relation $\sim$ on $\spec(U)$, whose equivalence classes we call \emph{blocks}. Two eigenvalues lie in the same block iff their best rational approximation with denominator $N_d$ is the same. If $\lambda,\lambda'$ lie in different blocks, then necessarily $|\lambda-\lambda'|>\tfrac{2}{3N_d}$.

In terms of the associated numerators,\eq{maineq2} and\eq{approx} imply
$$
\left|\frac{3\mu_0 - 2\mu_1}{N_d}\right| \leq \left| 3\left(\frac{\mu_0}{N_d} - \lambda_0\right) - 2\left(\frac{\mu_1}{N_d} - \lambda_1\right) + \left(3\lambda_0 - 2\lambda_1\right) \right| < \frac{1}{2N_d} + \frac{1}{3N_d} + \eps < \frac{1}{N_d} \:.
$$
Since $\mu_0$ and $\mu_1$ are integers (modulo $N_d$) and $N_d$ is odd, this makes $\mu_1$ unique (modulo $N_d$). So for a given $\lambda_0\in\spec(U)$ there is a unique block of eigenvalues $\lambda_1$ which can possibly satisfy\eq{maineq2}. For any $\lambda\in\spec(U)$, we also write $\tfrac{3}{2}\lambda$ for any $\lambda_1$ obtained in this way from $\lambda_0=\lambda$. 

We now prove the inequality\eq{cons}. Let $U=\sum_{\lambda\in\spec(U)}e^{2\pi i\,\lambda}P_\lambda$ be the spectral decomposition of $U$. We begin by bounding $||P_{\lambda'}VP_\lambda||$ for $\lambda,\lambda'\in\spec(U)$. The assumption $||VU^3-U^2V||<\eps$ implies
$$
||P_{\lambda'}V\cdot e^{2\pi i\,\cdot 3\lambda}P_\lambda - P_{\lambda'}e^{2\pi i\,\cdot 2\lambda'}VP_\lambda||<\eps \:,
$$
so that
$$
||P_{\lambda'}VP_\lambda||< \frac{\eps}{|e^{2\pi i\,\cdot 3\lambda} - e^{2\pi i\,\cdot 2\lambda'}|} = \frac{\eps}{\left|e^{2\pi i\,(3\lambda-2\lambda')}-1\right|}\:.
$$
When $\frac{3}{2}\lambda\not\sim \lambda'$, i.e. $\frac{3}{2}\lambda$ and $\lambda'$ do not lie in the same block, then
$$
|3\lambda - 2\lambda'| = |3\lambda - 2\cdot\tfrac{3}{2}\lambda + 2\cdot\tfrac{3}{2}\lambda - 2\lambda'|\geq 2|\tfrac{3}{2}\lambda - 2\lambda'| - |3\lambda - 2\cdot\tfrac{3}{2}\lambda| > \frac{2}{3N_d} - \eps > \frac{1}{3N_d}.
$$
In this case,
\beq
\label{es1}
||P_{\lambda'}VP_\lambda|| < \frac{\eps}{\left|e^{2\pi i\, \frac{1}{4N_d}}-1\right|} < 3N_d\eps \:.
\eeq
An analogous argument shows that the same estimate holds for $V^{-1}$: when $\frac{3}{2}\lambda\not\sim \lambda'$, then
\beq
\label{es2}
||P_{\lambda}V^{-1}P_{\lambda'}|| < 3N_d\eps \:.
\eeq
Now the desired quantity can be estimated as
\begin{align*}
||UV^{-1}& UV-V^{-1}UVU|| \\[.2cm]
& =\left|\left| \sum_{\lambda_1,\lambda_2\in\spec(U)} e^{2\pi i\,(\lambda_1+\lambda_2)} P_{\lambda_1}V^{-1}P_{\lambda_2}V - \sum_{\lambda_2,\lambda_3\in\spec(U)} e^{2\pi i\,(\lambda_2+\lambda_3)}V^{-1}P_{\lambda_2} VP_{\lambda_3} \right|\right| \\[.2cm]
& =\left|\left|\sum_{\lambda_1,\lambda_2,\lambda_3}\left(e^{2\pi i\,(\lambda_1+\lambda_2)}-e^{2\pi i\,(\lambda_2+\lambda_3)}\right)P_{\lambda_1}V^{-1}P_{\lambda_2}VP_{\lambda_3}\right|\right| \\[.2cm]
& \leq \sum_{\lambda_1,\lambda_2,\lambda_3}\left|\left|\left(e^{2\pi i\,\lambda_1}-e^{2\pi i\,\lambda_3}\right)P_{\lambda_1}V^{-1}P_{\lambda_2}VP_{\lambda_3}\right|\right| \\[.2cm]
& <  \sum_{\substack{\lambda_1,\lambda_2,\lambda_3\\ \textrm{ s.t. }\frac{3}{2}\lambda_1\sim \lambda_2\sim\frac{3}{2}\lambda_3}} \left|e^{2\pi i\,(\lambda_1-\lambda_3)}-1\right| + 2d^3\cdot 3N_d\eps \:.
\end{align*}
where the last estimate follows from\eq{es1},\eq{es2}. From $\frac{3}{2}\lambda_1\sim \lambda_2\sim\frac{3}{2}\lambda_3$, we conclude $\lambda_1\sim\lambda_3$, and hence $|\lambda_1-\lambda_3|<2\cdot 3^d d\eps$ by\eq{approx}. This lets us bound each summand by $2\pi \cdot 2\cdot 3^dd\eps$, so that
$$
||UV^{-1}UV-V^{-1}UVU|| <  d^3 ( 4\pi\cdot 3^d d + 3N_d ) \eps
$$
Since $4\pi\cdot 3^d d < N_d$ for $d\geq 3$, the assertion follows.
\end{proof}

\end{document}